\documentclass[11pt,a4paper]{article}
\usepackage[utf8]{inputenc}

\usepackage{fullpage} 

\usepackage{graphicx}
\usepackage{amsmath}
\usepackage{mathtools}
\usepackage{amsthm}
\usepackage{amssymb}
\usepackage{color}
\usepackage{booktabs}
\usepackage{graphicx}
\usepackage{float}
\usepackage{subfigure} 
\usepackage[square,comma,numbers]{natbib}
\usepackage{epstopdf}
\usepackage{dcolumn}
\usepackage[T1]{fontenc}
\usepackage[disable]{todonotes}
\usepackage{bm}
\usepackage{nicefrac}
\usepackage{esint}
\usepackage{thmtools}
\usepackage{thm-restate}

\usepackage{xspace}

\usepackage[algo2e,ruled,vlined,linesnumbered,algosection,hanginginout]{algorithm2e}
\SetEndCharOfAlgoLine{}

\usepackage{booktabs}
\usepackage{tabularx}
\usepackage{multirow}
\usepackage{makecell}
\usepackage[inline,shortlabels]{enumitem}
\setlist{nosep}
\usepackage{tikz}

\usepackage{pgfplots}
\usepackage{siunitx}
\usepackage{pgfplotstable}
\pgfplotsset{compat=1.5} 

\usepackage[pagebackref]{hyperref}
\usepackage[capitalize,nameinlink,noabbrev]{cleveref}

\usetikzlibrary{arrows,decorations.pathmorphing,decorations.pathreplacing,backgrounds,positioning,fit,matrix}
\usetikzlibrary{shapes,calc,patterns}

\tikzstyle{normalline} = [line width=.8pt]

\pgfplotsset{
    discard if not/.style 2 args={
        x filter/.code={
            \edef\tempa{\thisrow{#1}}
            \edef\tempb{#2}
            \ifx\tempa\tempb
            \else
                
            \fi
        }
    }
}

\tikzstyle{vert}=[circle,inner sep=1.5,fill=white,draw,minimum size=.3cm]
\tikzstyle{vert2}=[circle,inner sep=1.5,fill=white,draw,minimum size=.2cm]

\pgfdeclarelayer{bg} 
\pgfsetlayers{bg,main}

\usepackage{bm}
\usepackage{etoolbox}

\definecolor{ourblue}{RGB}{130,200,250}
\definecolor{ourdarkblue}{RGB}{30,60,140}
\definecolor{ourgreen}{RGB}{0,100,0}
\definecolor{ourred}{RGB}{180,20,30}
\definecolor{lipicsyellow}{RGB}{220,200,0}
\definecolor{ourorange}{RGB}{230,100,0}
\definecolor{ourgray}{RGB}{100,100,100}

\theoremstyle{plain}
\newtheorem{theorem}{Theorem}

\newtheorem{observation}[theorem]{Observation}
\newtheorem{proposition}[theorem]{Proposition}

\theoremstyle{definition}
\newtheorem{definition}[theorem]{Definition}

\theoremstyle{remark}
\crefname{equation}{Equality}{Equalities}
\Crefname{equation}{Inequality}{Inequalities}
\crefname{section}{Section}{Sections}
\crefname{subsection}{Section}{Sections}
\crefname{subsubsection}{Section}{Sections}
\crefname{theorem}{Theorem}{Theorems}
\crefname{observation}{Observation}{Observations}
\crefname{proposition}{Proposition}{Propositions}
\crefname{definition}{Definition}{Definitions}
\crefname{corollary}{Corollary}{Corollaries}
\crefname{lemma}{Lemma}{Lemmata}
\crefname{claim}{Claim}{Claims}
\crefname{figure}{Figure}{Figures}
\crefname{table}{Table}{Tables}
\crefname{algocf}{Algorithm}{Algorithms}
\Crefname{algocf}{Algorithm}{Algorithms}
\crefname{algorithm}{Algorithm}{Algorithms}
\crefname{example}{Example}{Examples}
\crefname{reduction}{Reduction}{Reductions}
\crefname{part}{Part}{Parts}
\crefname{line}{Line}{Lines}
\crefname{appendix}{Appendix}{Appendices}
\crefname{chapter}{Chapter}{Chapters}
\crefname{footnote}{Footnote}{Footnotes}

\newcommand{\classNP}{\text{NP}\xspace}

\newcommand{\NP}{\text{NP}\xspace}

\newcommand{\classcoNP}{\text{coNP}\xspace}
\newcommand{\FPT}{\text{FPT}\xspace}

\newcommand{\wx}[1]{\text{W[#1]}\xspace}
\newcommand{\wone}{\wx{1}}
\newcommand{\wtwo}{\wx{2}}

\newcommand{\polyadvice}{{\text{poly}}\xspace}

\newcommand{\NoKernelAssume}{\classNP $\subseteq$ \classcoNP/\polyadvice}

\newcommand{\N}{\ensuremath{\mathbb{N}}}

\newcommand{\truevalue}{\textrm{true}}
\newcommand{\falsevalue}{\textrm{false}}

\newcommand{\yes}{YES}
\newcommand{\no}{NO}

\newcommand{\proofpar}[1]{\smallskip\noindent\emph{#1}}

\newcommand{\lifetime}{\ensuremath{\ell}}

\newcommand{\TG}{\ensuremath{\mathcal{G}}\xspace}
\newcommand{\TGfull}{\ensuremath{\mathcal{G}=(V, E_1, E_2, \ldots, E_\lifetime)}\xspace}
\newcommand{\TGcompact}{\ensuremath{\mathcal{G}=(V, (E_i)_{i\in[\lifetime]})}\xspace}

\newcommand{\layer}{layer\xspace}

\newcommand{\timestep}{time step\xspace}

\newcommand{\RArrow}{\smallskip$(\Rightarrow)$:\xspace}
\newcommand{\LArrow}{\smallskip$(\Leftarrow)$:\xspace}

\newcommand{\nonstrpath}[1]{temporal~$(#1)$-path}

\newcommand{\nonstrsep}[1]{temporal~$(#1)$-sep\-a\-ra\-tor}
\newcommand{\nonstrseps}[1]{temporal~$(#1)$-sep\-a\-ra\-tors}

\newcommand{\probDef}[3]{
\begin{center}   
    \fbox{~\begin{minipage}{.95\textwidth}
      \vspace{2pt} 
     
      \noindent
      \normalsize\textsc{#1}
      
      \vspace{4pt}
      \setlength{\tabcolsep}{3pt}
      \renewcommand{\arraystretch}{1.0}
      \begin{tabularx}{\textwidth}{@{}lX@{}}
	\normalsize\emph{Input:} 	& \normalsize#2 \\
	\normalsize\emph{Question:} 	& \normalsize#3
      \end{tabularx}
    \end{minipage}}
    \end{center}
}

\newcommand{\probSeparator}{\textsc{Temporal $(s,z)$-Sep\-a\-ra\-tion}\xspace}

\newcommand{\probRSeparator}{\textsc{Restless Temporal $(s,z)$-Sep\-a\-ra\-tion}\xspace}

\newcommand{\probDefRSeparator}{
\probDef{\probRSeparator}
{A temporal graph~$\TGcompact$, two distinct vertices~$s,z\in V$, and two integers~$k \in \N$ and $\Delta\le\lifetime$.}
{Does $\TG$ admit a $\Delta$-restless \nonstrsep{s,z} of size at most~$k$?   }}

\newcommand{\probRestlessPath}{\textsc{Restless Temporal $(s,z)$-Path}\xspace}

\newcommand{\wait}{\ensuremath{\Delta}}

\usepackage{authblk}

\title{The Complexity of Finding Temporal Separators under Waiting Time Constraints
\footnote{This work is based on a previously unpublished chapter of the author's PhD-thesis~\cite{Molter20}.}}

\author{Hendrik Molter\thanks{Supported by the DFG,
project MATE (NI 369/17), and by the ISF, grant No.~1070/20. Main part of this work was done while the author was affiliated with TU~Berlin.}}

\date{ }

\affil{Department of Industrial Engineering and Management, Ben-Gurion~University~of~the~Negev, 
Beer-Sheva, 
Israel, 
\texttt{molterh@post.bgu.ac.il}}

\begin{document}

%
%
%
%
%
%
%

\maketitle

\begin{abstract}
In this work, we investigate the computational complexity of \probRSeparator, where we are asked whether it is possible to destroy all restless temporal paths between two distinct vertices~$s$ and~$z$ by deleting at most $k$ vertices from a temporal graph.
A temporal graph has a fixed vertex  but the edges have (discrete) time stamps.  A restless temporal path uses edges with non-decreasing time stamps and the time spent at each vertex must not exceed a given duration $\Delta$.

    \probRSeparator naturally generalizes the \NP{}-hard \probSeparator problem. We show that \probRSeparator is complete for $\Sigma_2^\text{P}$, a complexity class located in the second level of the polynomial time hierarchy. We further provide some insights in the parameterized complexity of \probRSeparator parameterized by the separator size $k$.
\end{abstract}

\section{Introduction}
Capturing dynamic changes in a network plays an increasingly important role in network analysis and algorithmics and temporal graphs are a popular model that is able to represent such changes over time~\cite{Cas+12,Hol15,HS19,LVM18,Mic16}. Especially the notion of connectivity is much more intricate in the temporal setting and path-related problems were among the first ones studied on temporal graphs~\cite{Ber96,KKK02}. In particular, vertex separators are NP-hard to find in temporal graphs~\cite{KKK02} whereas it is possible to find them in polynomial time in static graphs~\cite[Theorem~6.8]{AMO93}.

In this work we study the computational complexity of finding vertex separators in a temporal graph that destroy all temporal paths that obey certain \emph{waiting time constraints}. We call such paths \emph{restless temporal paths} and their ``waiting'' or ``pausing'' time at a vertex is restricted to some prescribed duration $\Delta$. Restless temporal paths naturally model infection transmission chains of diseases that grant immunity upon recovery~\cite{Hol16}. Such transmission routes are captured by the well-established \emph{SIR-model} (Susceptible-Infected-Recovered)~\cite{Bar16,kermack1927contribution,New18}. This also motivates the search of vertex separators since they naturally model breaking infection transmission by vaccinations. This work is focused on the analysis of the computational complexity of \probRSeparator, the problem of deciding whether a temporal graph admits a separator of size at most $k$ that destroys all restless temporal paths between two designated vertices $s$ and $z$ (a formal problem definition is given in \cref{sec:sep:prelims}).

\paragraph{Related Work.}

The problem of finding minimum temporal separators was first studied by \citet{KKK02} and they proved it is \NP-hard.
In contrast, \citet{Ber96} proved earlier that destroying all temporal path between two designated vertices by deleting a minimum number of edges instead of vertex can be done in polynomial time. 
\citet{Zsc+19} and \citet{Flu+19a} further studied the computational complexity of finding temporal separators of bounded size and provide (parameterized) algorithms as well as hardness result for several restricted cases. \citet{maack21} studied the problem on specific temporal graph classes.

The computation of restless temporal paths has been studied by \citet{HMZ19}. They show that deciding whether a restless temporal path exists between two vertices is NP-complete and they give several further (parameterized) hardness and algorithmic results. Notably, finding restless temporal \emph{walks} between two vertices in known to be polynomial-time solvable~\cite{Him+20}.

\paragraph{Our Contributions.}


We analyze the computational complexity of \probRSeparator and show that this problem is $\Sigma_2^\text{P}$-complete. We further give some insights on the parameterized complexity of \probRSeparator parameterized by the separator size.

\section{Preliminaries and Basic Observations}\label{sec:sep:prelims}

In this section, we formally introduce the most important concepts related to restless temporal separators and give the formal problem definition of \probRSeparator. We further discuss some basic observations.

\paragraph{Computational Complexity.}
In this work we show that the problem under consideration is complete for the complexity class $\Sigma_2^\text{P}$~\cite{AB09,Sto76}. This complexity class is located in the second level of the polynomial time hierarchy and contains both NP and coNP.
It is closed under polynomial-time many-one reductions and, intuitively, contains all problems that are at most as hard as the problem \textsc{$\exists\forall$-SAT}~\cite{AB09,Sto76}, where we are given a Boolean formula $\phi$ in conjunctive normal form and the variables of $\phi$ are partitioned into two sets $X$ and $Y$, and we are asked to decide whether there exists an assignment for all variables in $X$ such that for all possible assignments for the variables in $Y$, the formula $\phi$ evaluates to \truevalue. The problem \textsc{$\exists\forall$-SAT} is complete for $\Sigma_2^\text{P}$~\cite{AB09,Sto76}. The class $\Sigma_2^\text{P}$ can also be characterized as the set of all problems that can be solved by an NP-machine that has oracle access to an NP-complete problem~\cite{AB09,Sto76}.

We use standard notation and terminology from parameterized
complexity theory~\cite{Cyg+15,DF13,FG06,Nie06}
and give here a brief overview of the most important concepts that are used in this paper.
A \emph{parameterized problem} is a language $L\subseteq \Sigma^* \times \N$, where $\Sigma$ is a finite alphabet. We call the second component
the \emph{parameter} of the problem.
A parameterized problem is in the complexity class XP if there is an
algorithm that solves each instance~$(I,r)$ in~$|I|^{f(r)}$ time, for some
computable function $f$.
A parameterized problem is \emph{fixed-pa\-ram\-e\-ter tractable} (i.e., in the complexity class \FPT{})
if there is an algorithm that solves each instance~$(I, r)$ in~$f(r) \cdot |I|^{O(1)}$ time,
for some computable function~$f$. 
A decidable parameterized problem $L$ admits a \emph{polynomial kernel} if there is a polynomial-time algorithm that transforms each instance $(I,r)$ into an instance $(I', r')$ such that $(I,r)\in L$ if and only if $(I',r')\in L$ and $|(I', r')|\in r^{O(1)}$. 
If a parameterized problem is hard for the parameterized complexity class \wone or \wtwo\, then it is (presumably) not in~\FPT{}.
The complexity classes \wone\ and \wtwo\ is closed under parameterized reductions, which may run in \FPT-time and additionally set the new parameter to a value that exclusively depends on the old parameter. They are part of the W-hierarchy and it is known that $\text{FPT}\subseteq \wone \subseteq \wtwo \subseteq \ldots \subseteq \text{XP}$.
If a parameterized problem is \NP-hard (resp.\ coNP-hard) for constant parameter values, then the problem is para-NP-hard (resp.\ para-coNP-hard).

\paragraph{Temporal Graphs, Paths, and Separators.}
An (undirected, simple) \emph{temporal graph} is a tuple~$\TGfull$ (or \TGcompact for short), with $E_i\subseteq\binom{V}{2}$ for all $i\in[\lifetime]$.
We call $\lifetime(\TG) := \lifetime$ the \emph{lifetime} of $\TG$. 
We call the graph $G_i(\TG) = (V, E_i(\TG))$ the \emph{\layer}~$i$ of $\TG$
where $E_i(\TG) := E_i$. If $E_i=\emptyset$, then $G_i$ is a \emph{trivial} layer.
We call layers $G_i$ and $G_{i+1}$ \emph{consecutive}.
We call $i$ a \emph{\timestep}. If an edge $e$ is present at time $i$, that is, $e\in E_i$, we say that $e$ has \emph{time stamp} $i$ and call the pair $(e,i)$ a \emph{time edge}.. We further denote $V(\TG):=V$. %
%

A \emph{restless} temporal path is not allowed to wait an arbitrary amount of time in a vertex, but has to leave any vertex it visits within the next $\Delta$-window, for some given value for $\Delta$. Formally, they are defined as follows.
\begin{definition}[Restless Temporal Walk / Restless Temporal Path]\label{def:rtemppath}
	A \emph{$\wait$-restless temporal walk} of 
length~$n$ from vertex $s$ to vertex $z$ in a temporal graph~$\TGcompact$ is a sequence
$P =\big( (s=v_0,v_1, t_1), (v_1, v_2, t_2), \dots, (v_{n-1}, v_n=z, t_{n} ) \big)$ such that for all $i\in[n]$ we have that $\{v_{i-1},v_i\}\in E_i$ and for all $i\in [n-1]$ we have that $t_i \leq t_{i+1} \leq t_i + \wait$.
Moreover, we call $P$ a \emph{$\Delta$-restless temporal path} of length~$n$ 
if~$v_i\neq v_j$ for all~$i, j\in \{0,\ldots,n\}$ with $i\neq j$.
We say that $P$ \emph{respects} the maximum waiting time $\wait$.
\end{definition}
We call the problem of checking whether there exists a restless tempora $(s,z)$-path in a given temporal graph for a given $\Delta$ value ``\probRestlessPath''. This problem is known to be NP-hard and has been thoroughly investigated by \citet{HMZ19}.

Now we are ready to give the formal definition of a restless temporal $(s,z)$-separator, which should destroy all restless temporal $(s,z)$-paths in a given temporal graph.

\begin{definition}[$\Delta$-Restless Temporal $(s,z)$-Separator]
Let $\TGcompact$ be a temporal graph with $s,z\in V$. Let $\Delta\le\lifetime$. A vertex set~$S\subseteq V\setminus\{s,z\}$ is a \emph{$\Delta$-restless \nonstrsep{s,z}} for $\TG$ if there is no $\Delta$-restless \nonstrpath{s,z} in~$\TG - S$.
\end{definition}

We can now formally define the (decision) problem of finding a $\Delta$-restless \nonstrsep{s,z} in a given temporal graph $\TG$ with two distinct vertices $s$ and~$z$. 

\probDefRSeparator

\paragraph{Basic Observations.}
Since \probRSeparator generalizes the NP-hard \probSeparator problem~\cite{Flu+19a,KKK02,Zsc+19}, we can observe that the problem is clearly \NP-hard. 
However, we can also observe that we presumably cannot verify in polynomial time whether a vertex set $S$ is a $\Delta$-restless \nonstrsep{s,z} for a given temporal graph $\TG$. \citet{HMZ19} showed checking whether there is a $\Delta$-restless temporal path from $s$ to $z$ in $\TG-S$ is \NP-hard. Note that \probRSeparator with $k=0$ is the complement of \probRestlessPath. Hence we can observe the following. 
\begin{observation}\label{obs:paraconp}
\probRSeparator is coNP-hard for all $\Delta\ge 1$ even if~$k=0$.
\end{observation}
This implies that \probRSeparator is presumably \emph{not} contained in~\NP.

Furthermore, it is easy to observe that computational hardness of \probRestlessPath for some fixed value of $\Delta$ implies hardness for all larger values of $\Delta$.
This allows us to construct hardness reductions for small fixed values of~$\Delta$ and still obtain general hardness results. The proof of this observation is essentially the same as the proof of an analogous result for \probRestlessPath by \citet{HMZ19}.
\begin{observation}\label{obs:rsep:delta}
	For every fixed $\Delta$, \probRSeparator on instances $(\TG, s,z,k,\Delta+1)$ is computationally at least as hard as 
	\probRSeparator on instances $(\TG, s,z,k,\Delta)$.
\end{observation}
\begin{proof}
	The result immediately follows from the observation that a temporal graph~$\TG$ contains a $\Delta$-restless \nonstrpath{s,z} if and only if the temporal graph $\TG'$ contains a $(\Delta+1)$-restless \nonstrpath{s,z}, where
	$\TG'$ is obtained from $\TG$ by inserting one trivial (that is, edgeless) layer
	after every $\Delta$ consecutive layers.
\end{proof}

Finally, we can also observe that \probRSeparator is fixed-parameter tractable when parameterized by the number $|V|$ of vertices. We can check for each vertex set size $k$ whether it is $\Delta$-restless \nonstrsep{s,z}. We remove it from the input temporal graph and then use an FPT-algorithm for \probRestlessPath when parameterized by the number $|V|$ of vertices~\cite{HMZ19} to verify whether there is no $\Delta$-restless \nonstrpath{s,z}.
\begin{observation}
\probRSeparator parameterized by the number $|V|$ of vertices is fixed-parameter tractable. 
\end{observation}
However, we presumably cannot obtain a polynomial kernel for the parameter number $|V|$ of vertices since we can observe the following kernelization lower bound for \probRSeparator.
\begin{observation}
\probRSeparator parameterized by the number $|V|$ of vertices does not admit a polynomial kernel for all $\Delta\ge 1$ unless \NoKernelAssume. 
\end{observation}
This follows directly from the result by \citet{HMZ19} that \probRestlessPath does not admit a polynomial kernel when parameterized by the number $|V|$ of vertices for all $\Delta\ge 1$ unless \NoKernelAssume. This follows from the observation that \probRSeparator with $k=0$ is the complement of \probRestlessPath and hence a polynomial kernel for \probRSeparator parameterized by $|V|$ would also be a polynomial kernel for \probRestlessPath parameterized by $|V|$.

\section{Computational Complexity of Restless Temporal Separators}\label{sec:restlesssep}

In this section we investigate the computational complexity of \probRSeparator. The fact that the problem is both NP-hard and coNP-hard as observed in the previous section already suggests that \probRSeparator is located somewhere higher in the polynomial time hierarchy. Indeed we can show that \probRSeparator is $\Sigma^\text{P}_2$-complete. This implies, for example, that we presumably cannot use SAT-solvers or ILP-solvers to compute $\Delta$-restless \nonstrseps{s,z}. To show $\Sigma^\text{P}_2$-hardness, we give a reduction from \textsc{$\exists\forall$-SAT}, where we are given a Boolean formula $\phi$ in conjunctive normal form and a partition of variables of $\phi$ into two sets $X$ and $Y$. Then we are asked to decide whether there exists an assignment for all variables in $X$ such that for all possible assignments for the variables in $Y$, the formula $\phi$ evaluates to \truevalue. The very rough idea of our reduction is that the vertices selected for the separator correspond to an assignment for the variables in $X$ and if there is an assignment for the variables in $Y$ such that $\phi$ evaluates to \falsevalue, then the temporal graph should still contain a $\Delta$-restless \nonstrpath{s,z} after the separator vertices are removed.

\begin{theorem}\label{thm:sigmap2}
\probRSeparator is $\Sigma^\text{P}_2$-complete for all $\Delta\ge 1$ even if every edge has only one time stamp.
\end{theorem}
\begin{proof}
We present a polynomial-time reduction from the $\Sigma^\text{P}_2$-complete problem \textsc{$\exists\forall$-SAT}~\cite{AB09,Sto76}, where we are given a Boolean formula $\phi$ in conjunctive normal form and the variables of $\phi$ are partitioned into two sets $X$ and $Y$, and we are asked to decide whether there exists an assignment for all variables in $X$ such that for all possible assignments for the variables in $Y$, the formula $\phi$ evaluates to \truevalue.

Let $\phi(X,Y)$ denote an instance of \textsc{$\exists\forall$-SAT}, let $n_X=|X|$, $n_Y=|Y|$, and let $m$ be the number of clauses in $\phi$. We assume that no clause of $\phi$ contains a variable several times. We construct a temporal graph $\TGcompact$ with $\lifetime=2m+1$, consisting of three gadgets. We start with the ``exists gadget'' in which we have to select the vertices of the $\Delta$-restless \nonstrsep{s,z}. Intuitively, this chooses an assignment for the variables in $X$. The next gadget is the ``forall gadget'' which must be passed by every $\Delta$-restless \nonstrpath{s,z}. This gadget can be traversed in $2^{n_Y}$ ways which, intuitively, represent all possible assignments for the variables in $Y$. The last gadget is the clause gadget which, intuitively, a $\Delta$-restless \nonstrpath{s,z} can only pass if there is a clause that is not satisfied. We set $\Delta=1$ and $k=n_X$. We next give formal descriptions of the gadgets.

\proofpar{Exists Gadget.} We start by creating two vertices $s$ and $z$. For every variable $x_i\in X$ we create two vertices $x^{(T)}_i$ and $x^{(F)}_i$ and we add edges $\{s,x^{(T)}_i\}$, $\{x^{(T)}_i,x^{(F)}_i\}$, and $\{x^{(F)}_i,z\}$ to~$E_1$. This already completes the construction of the exists gadget. We can see that we created~$n_X$ internally vertex-disjoint $\Delta$-restless \nonstrpath{s,z}s. Since we set~$k=n_X$, we have that every $\Delta$-restless \nonstrsep{s,z} has to contain one vertex from each of these paths.

\proofpar{Forall Gadget.} For every variable $y_i\in Y$ we create two vertices $y^{(T)}_i$ and $y^{(F)}_i$. We further create $n_Y+1$ vertices $s_1, s_2, \ldots s_{n_Y+1}$. For all $i\in[n_Y]$ we add edges $\{s_i,y^{(T)}_i\}$, $\{s_i,y^{(F)}_i\}$, $\{y^{(T)}_i,s_{i+1}\}$, and $\{y^{(F)}_i,s_{i+1}\}$ to $E_1$. We further add edge $\{s,s_1\}$ to $E_1$. This completes the construction of the forall gadget. We can see that there are $2^{n_Y}$ different $\Delta$-restless temporal paths from $s_1$ to $s_{n_Y+1}$. Intuitively, each one of these represents an assignment for the variables in $Y$.

\proofpar{Clause Gadget.} For every clause $c_i$ of $\phi$ we create two vertices $c_i^{(1)}$ and $c_i^{(2)}$. For every~$i\in[m]$ we add edge $\{c_i^{(1)},c_i^{(2)}\}$ to $E_{2i+1}$ and if $i<m$, then we add edge $\{c_i^{(2)},c_{i+1}^{(1)}\}$ to~$E_{2i+2}$. We further add edge $\{s_{n_Y+1},c_1^{(1)}\}$ to $E_2$. We call this part of the gadget the \emph{clause selection path}.

Let $c_i$ be a clause for some $i\in [m]$. Without loss of generality let $c_i$ contain variables $x_1,\ldots, x_{j_1}$ and $y_1,\ldots,y_{j_2}$. Then we add the following edges to $E_{2i+1}$: 
\begin{itemize}
\item If $x_1$ appears non-negated in $c_i$, then we add edge $\{c_i^{(2)},x_1^{(F)}\}$, otherwise we add edge $\{c_i^{(2)},x_1^{(T)}\}$.
\item For all $j\in[j_1-1]$, if $x_j$ appears non-negated in $c_i$, then set $v_j=x_j^{(F)}$, otherwise set $v_j=x_j^{(T)}$. If $x_{j+1}$ appears non-negated in $c_i$, then set $v_{j+1}=x_{j+1}^{(F)}$, otherwise set $v_{j+1}=x_{j+1}^{(T)}$. We add edge $\{v_j,v_{j+1}\}$.
\item If $x_{j_1}$ appears non-negated in $c_i$, then set $v=x_{j_1}^{(F)}$, otherwise set $v=x_{j_1}^{(T)}$. If $y_1$ appears non-negated in $c_i$, then set $w=y_1^{(F)}$, otherwise set $w=y_1^{(T)}$. We add edge $\{v,w\}$.
\item For all $j\in[j_2-1]$, if $y_j$ appears non-negated in $c_i$, then set $v_j=y_j^{(F)}$, otherwise set $v_j=y_j^{(T)}$. If $y_{j+1}$ appears non-negated in $c_i$, then set $v_{j+1}=y_{j+1}^{(F)}$, otherwise set $v_{j+1}=y_{j+1}^{(T)}$. We add edge $\{v_j,v_{j+1}\}$.
\item If $y_{j_2}$ appears non-negated in $c_i$, then we add edge $\{y_{j_2}^{(F)},z\}$, otherwise we add edge $\{y_{j_2}^{(T)},z\}$.
\end{itemize}
We do this for all clauses in $\phi$. This completes the construction of the clause gadget. Intuitively, a $\Delta$-restless \nonstrpath{s,z} should only be able to traverse the clause gadget if there is a clause that is not satisfied.

This finishes the construction of $\TGcompact$. The construction is illustrated in \cref{fig:sigmap2}. Recall that $\Delta=1$ and $k=n_X$. It is easy to check that $\TG$ can be constructed in polynomial time and that every edge has at most one time stamp.

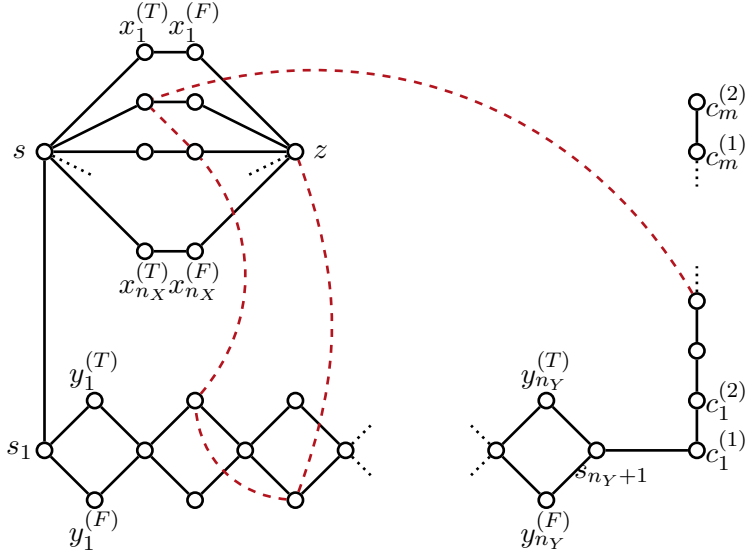
\begin{figure}[t]
\begin{center}
\begin{tikzpicture}[line width=1pt, scale=.33, yscale=1]
    \node (A1l) at (0,1.2) {$x_1^{(T)}$}; 
    \node (A2l) at (2,1.2) {$x_1^{(F)}$}; 
    \node[vert2] (A1) at (0,0) {}; 
    \node[vert2] (A2) at (2,0) {}; 
    \node[vert2] (B1) at (0,-2) {}; 
    \node[vert2] (B2) at (2,-2) {};
    \node[vert2] (C1) at (0,-4) {}; 
    \node[vert2] (C2) at (2,-4) {};
    \node[vert2] (D1) at (0,-8) {}; 
    \node[vert2] (D2) at (2,-8) {};
    \node (D1l) at (0,-9.2) {$x_{n_X}^{(T)}$}; 
    \node (D2l) at (2,-9.2) {$x_{n_X}^{(F)}$};
    
    \node[vert2] (S) at (-4,-4) {}; 
    \node[vert2] (Z) at (6,-4) {};
    \node (Sl) at (-5,-4) {$s$}; 
    \node (Zl) at (7,-4) {$z$};
    
    \node[vert2] (S1) at (-4,-16) {};
    \node (S1l) at (-5,-16) {$s_1$};
    \node (V11l) at (-2,-12.8) {$y_1^{(T)}$};
    \node (V12l) at (-2,-19.2) {$y_1^{(F)}$}; 
    \node[vert2] (V11) at (-2,-14) {};
    \node[vert2] (V12) at (-2,-18) {}; 
    \node[vert2] (S2) at (0,-16) {};
    \node[vert2] (V21) at (2,-14) {};
    \node[vert2] (V22) at (2,-18) {}; 
    \node[vert2] (S3) at (4,-16) {};
    \node[vert2] (V31) at (6,-14) {};
    \node[vert2] (V32) at (6,-18) {}; 
    \node[vert2] (S4) at (8,-16) {};
    
    \node[vert2] (S5) at (14,-16) {};
    \node[vert2] (V51) at (16,-14) {};
    \node[vert2] (V52) at (16,-18) {}; 
    \node (V51l) at (16,-12.8) {$y_{n_Y}^{(T)}$};
    \node (V52l) at (16,-19.2) {$y_{n_Y}^{(F)}$}; 
    \node[vert2] (S6) at (18,-16) {};
    \node (S6l) at (18.6,-17) {$s_{n_Y+1}$};
    
    \node[vert2] (CC1) at (22,-16) {};
    \node[vert2] (CC2) at (22,-14) {};
    \node[vert2] (CC3) at (22,-12) {};
    \node[vert2] (CC4) at (22,-10) {};
    \node[vert2] (CC5) at (22,-4) {};
    \node[vert2] (CC6) at (22,-2) {};
    \node (CC1l) at (23.2,-16) {$c_1^{(1)}$};
    \node (CC2l) at (23.2,-14) {$c_1^{(2)}$};
    \node (CC5l) at (23.2,-4) {$c_m^{(1)}$};
    \node (CC6l) at (23.2,-2) {$c_m^{(2)}$};
    
    \draw (S) -- (A1);
    \draw (S) -- (B1);
    \draw (S) -- (C1);
    \draw (S) -- (D1);
    \draw[dotted] (S) -- (-2,-5);
    \draw (Z) -- (A2);
    \draw (Z) -- (B2);
    \draw (Z) -- (C2);
    \draw (Z) -- (D2);
    \draw[dotted] (Z) -- (4,-5);
    \draw (A1) -- (A2);
    \draw (B1) -- (B2);
    \draw (C1) -- (C2);
    \draw (D1) -- (D2);
    
    \draw (S) -- (S1);
    \draw (S1) -- (V11);
    \draw (S1) -- (V12);
    \draw (S2) -- (V11);
    \draw (S2) -- (V12);
    \draw (S2) -- (V21);
    \draw (S2) -- (V22);
    \draw (S3) -- (V21);
    \draw (S3) -- (V22);
    \draw (S3) -- (V31);
    \draw (S3) -- (V32);
    \draw (S4) -- (V31);
    \draw (S4) -- (V32);
    \draw[dotted] (S4) -- (9,-17);
    \draw[dotted] (S4) -- (9,-15);
    \draw[dotted] (S5) -- (13,-17);
    \draw[dotted] (S5) -- (13,-15);
    \draw (S5) -- (V51);
    \draw (S5) -- (V52);
    \draw (S6) -- (V51);
    \draw (S6) -- (V52);
    
    \draw (S6) -- (CC1);
    \draw (CC1) -- (CC2);
    \draw (CC2) -- (CC3);
    \draw (CC3) -- (CC4);
    \draw[dotted] (CC4) -- (22,-8.6);
    \draw[dotted] (CC5) -- (22,-5.4);
    \draw (CC5) -- (CC6);
    
    \path[bend right=40,color=ourred,dashed] (CC4) edge (B1);
    \path[bend left=40,color=ourred,dashed] (V32) edge (V21);
    \path[bend right=40,color=ourred,dashed] (V21) edge (C2);
    \path[color=ourred,dashed] (C2) edge (B1);
    \path[bend right=20,color=ourred,dashed] (V32) edge (Z);
\end{tikzpicture}
    \end{center}
    \caption{Visualization of parts of the underlying graph of the temporal graph $\TG$ constructed in the reduction of \cref{thm:sigmap2}. The red dashed path corresponds to the clause gadget for clause $c_2=(\neg x_2\vee x_3\vee \neg y_2\vee y_3)$ (without the clause selection path).}\label{fig:sigmap2}
\end{figure}

\proofpar{Correctness.}
Now we show that $\TG$ admits a $\Delta$-restless \nonstrsep{s,z} of size at most~$k$ if and only if $\phi$ is a \yes-instance of \textsc{$\exists\forall$-SAT}.

\RArrow Assume that there is an assignment for the variables in $X$ such that for all assignments for the variables of $Y$ we have that $\phi$ evaluates to \truevalue. We construct a $\Delta$-restless \nonstrsep{s,z} $S$ for $\TG$ as follows. For each $i\in[n_X]$, if variable~$x_i$ is assigned the value \truevalue, then we add the vertex $x_i^{(T)}$ to~$S$, otherwise we add $x_i^{(F)}$ to~$S$. Clearly, we have that $|S|=n_X=k$. In the following, we show that $S$ is a $\Delta$-restless \nonstrsep{s,z} for $\TG$.

Assume for contradiction that there is a $\Delta$-restless \nonstrpath{s,z} $P$ in $\TG-S$. It is easy to see that all $\Delta$-restless \nonstrpath{s,z}s in $\TG$ that only use edges from the exists gadget are destroyed in $\TG-S$ since from every such path, we put one vertex into $S$. Observe that all time edges adjacent to $z$ that are not part of the exists gadget have a time stamp of three or larger. Hence, to reach a time edge with time step two,~$P$ has to pass the forall gadget to reach time edge $\{s_{n_Y+1},c_1^{(1)}\}$, which is the only time edge with time stamp two. From this it follows that for every $i\in[n_Y]$ we have that either~$y_i^{(T)}\in V(P)$ or $y_i^{(F)}\in V(P)$. Then the path~$P$ enters the clause selection path of the clause gadget. To reach $z$, the path $P$ has to leave this path at some vertex~$c_j^{(2)}$ for some~$j\in[m]$ (meaning that $c_j^{(2)}\in V(P)$ and $c_{j+1}^{(1)}\notin V(P)$). We claim that this implies that clause $c_j$ is not satisfied if the variables from $Y$ are assigned the following truth values: for each $i\in [n_Y]$, if $y_i^{(T)}\in V(P)$, then we set $y_i$ to \truevalue, otherwise we set~$y_i$ to \falsevalue. Assuming that $c_j^{(2)}\in V(P)$ and $c_{j+1}^{(1)}\notin V(P)$, the only way to reach $z$ from~$c_j^{(2)}$ is through vertices that correspond to the variables appearing in clause $c_i$, since the time stamps from all paths from the clause selection path to~$z$ differ by at least two. More specifically, for each variable $x_i$ ($y_i$) appearing in~$c_j$, we have that~$V(P)$ contains vertex~$x_i^{(F)}$~($y_i^{(F)}$) if~$x_i$ ($y_i$) appears non-negated in $c_j$ and $V(P)$ contains vertex~$x_i^{(T)}$~($y_i^{(T)}$) otherwise. This means for the variables~$x_i$ that they are set to truth values that do not satisfy clause~$c_j$, otherwise the corresponding vertices would be contained in the separator $S$. For the variables $y_i$ this means the assignment we constructed earlier also sets them to truth values that do not satisfy clause~$c_j$, otherwise the corresponding vertices would have been used by $P$ when the path was passing the forall gadget at time step one. Hence, we have found an assignment for the variables in $Y$ such that together with the given assignment for the variables in~$X$, the formula $\phi$ evaluates to \falsevalue---a contradiction.

\LArrow Let $S\subseteq V\setminus \{s,z\}$ with $|S|\le k$ be a $\Delta$-restless \nonstrsep{s,z} for $\TG$. Let us first look at the exists gadget of $\TG$. It consists of~$n_X$ internally vertex-disjoint $\Delta$-restless \nonstrpath{s,z}s, each one visiting four vertices: $s$, $x^{(T)}_i$, $x^{(F)}_i$, and~$z$ for some $i\in[n_X]$. Of each of these $\Delta$-restless \nonstrpath{s,z}s, one vertex other than~$s$ or~$z$ has to be contained in $S$, otherwise $S$ would not be a $\Delta$-restless \nonstrsep{s,z}. It follows that for all $i\in[n_X]$ either $x^{(T)}_i$ or $x^{(F)}_i$ is contained in $S$ (and also no other vertices are contained in $S$ since $k=n_X)$. This lets us construct an assignment for the variables in~$X$ as follows. For every~$i\in[n_X]$, if $x^{(T)}_i\in S$, then we set~$x_i$ to \truevalue, otherwise we set~$x_i$ to \falsevalue. We claim that using this assignment for the variables in $X$, we have that for all assignments for the variables in $Y$ the formula~$\phi$ evaluates to \truevalue.

Assume for the sake of contradiction that there is an assignment for the variables in $Y$ such that together with the constructed assignment for the variables in $X$, the formula~$\phi$ evaluates to \falsevalue. Then we can construct a $\Delta$-restless \nonstrpath{s,z} in~$\TG-S$ as follows. Starting from~$s$ we traverse the forall gadget as follows. Starting with~$i=1$ to~$n_Y$ we visit $s_i$ and then $y_i^{(T)}$ if $y_i$ is set to \truevalue, and $y_i^{(F)}$ otherwise. Then we visit~$s_{n_Y+1}$. Up until this point, the path only uses time edges with time stamp one. Since $\phi$ evaluates to false, there is at least one clause in $\phi$ that is not satisfied. Let $c_j$ be that clause. We continue our $\Delta$-restless temporal path from $s_{n_Y+1}$ to~$c_j^{(2)}$. Since~$c_j$ evaluates to \falsevalue, the vertices corresponding to the variables in $X$ appearing in~$c_j$ are not contained in $S$, otherwise, by construction, the clause~$c_j$ would evaluate to \truevalue. Similarly, the vertices corresponding to the variables in $Y$ appearing in $c_j$ have not been visited by the path when traversing the forall gadget. Hence, we can continue the $\Delta$-restless temporal path from $c_j^{(2)}$ to $z$---a contradiction.

\proofpar{Containment in $\Sigma^\text{P}_2$.} Our proof so far shows that \probRSeparator is $\Sigma^\text{P}_2$-hard. To show that the problem is $\Sigma^\text{P}_2$-complete, we show that it is contained in $\Sigma^\text{P}_2$. Recall that $\Sigma^\text{P}_2$ contains all problems that can be solved by an NP-machine that has oracle access to an NP-complete problem~\cite{AB09,Sto76}. We can solve an instance $\TGcompact,s,z,k,\Delta)$ of \probRSeparator with such a machine as follows. We non-deterministically guess a set $S\subseteq V$ of size at most~$k$ and then produce an instance $(\TG-S,s,z,\Delta)$ of \probRestlessPath. Since \probRestlessPath is contained in \NP, we can reduce it to the NP-complete problem to which we have oracle access. We use the reduction to produce an equivalent instance of the NP-complete problem we have oracle access to and query the oracle with this instance. If the oracle answers \no, then we have found a $\Delta$-restless \nonstrsep{s,z} of size at most $k$ for $\TG$ and can answer \yes. It is easy to see that the described machine has an accepting path if and only if the \probRSeparator instance is a \yes-instance.
\end{proof}

From a parameterized complexity perspective we can make one rather straightforward observation. Since \probRSeparator generalizes \probSeparator, we know that \probRSeparator parameterized by the separator size $k$ is \wone-hard~\cite{Zsc+19}. However, we can observe that \probRSeparator is even \wtwo-hard when parameterized by the separator size $k$ by a straightforward reduction from \textsc{Hitting Set}, where we model each element of the universe with a vertex and each set by a path through the corresponding vertices. The waiting time $\Delta$ allows us to obtain a one-to-one correspondence between restless temporal paths in the constructed temporal graph and sets in the \textsc{Hitting Set} instance. We remark that the reduction we use to show this result has been used in a very similar way by ~\citet{Zsc17} to show that finding temporal separators of bounded size that destroy all $\Delta$-restless temporal \emph{walks} from $s$ to $z$ is \wtwo-hard when parameterized by the bound on the separator size. 
\begin{proposition}\label{prop:rsep:wtwo}
\probRSeparator parameterized by the separator size~$k$ is \wtwo-hard for all $\Delta\ge 1$.
\end{proposition}
\begin{proof}
We present a parameterized polynomial-time reduction from \textsc{Hitting Set}, where we are given a universe set $U$, a family of sets $S_1, \ldots, S_m\subseteq U$, and an integer $h$, and are asked whether there is a \emph{hitting set} $S^\star\subseteq U$ with $|S^\star|\le h$ such that for all $i\in[m]$ we have that $S^\star\cap S_i\neq\emptyset$. \textsc{Hitting Set} is \wtwo-complete when parameterized by $h$~\cite{DF99,PM81}.

Given an instance $(U, (S_i)_{i\in[m]}, h)$ of \textsc{Hitting Set}, we construct a temporal graph $\TGcompact$ with $\lifetime=2m$ as follows. We set $V=U\cup\{s,z\}$ and for each set $S_i$ with $i\in[m]$ we create two layers $G_{2i-1}$ and $G_{2i}$. In layer $G_{2i-1}$ we create a path from $s$ to $z$ that visits all vertices in $S_i$ in an arbitrary order. The layer $G_{2i}$ is trivial. We set $\Delta=1$ and $k=h$. This finishes the construction. It is easy to check that this can be done in polynomial time.

\proofpar{Correctness.} The correctness is straightforward to see. A $\Delta$-restless \nonstrsep{s,z} for $\TG$ has to contain at least one vertex from each set $S_i$ with $i\in[m]$, otherwise there would be a layer that contains a $\Delta$-restless \nonstrpath{s,z}. It follows that a $\Delta$-restless \nonstrsep{s,z} for $\TG$ is a hitting set for $(U, (S_i)_{i\in[m]})$. For the other direction, one has to observe that due to the waiting time restriction $\Delta$ and the trivial layers that are present in $\TG$, each $\Delta$-restless \nonstrpath{s,z} in $\TG$ corresponds to a set $S_i$ for some $i\in[m]$ of the \textsc{Hitting Set} instance. It follows that a hitting set contains at least one vertex from each $\Delta$-restless \nonstrpath{s,z} in~$\TG$.
\end{proof}
We remark that it is open whether \probSeparator parameterized by the separator size $k$ is contained in \wone. Hence, \cref{prop:rsep:wtwo} does not necessarily imply that \probRSeparator parameterized by the separator size $k$ is harder than \probSeparator parameterized by the separator size $k$. We remark that containment of \probRSeparator parameterized by the separator size~$k$ in \wtwo is unlikely since \cref{obs:paraconp} shows that \probRSeparator parameterized by the separator size $k$ is also para-coNP-hard. We conjecture that \probRSeparator parameterized by the separator size $k$ is complete for the parameterized complexity class~$\Sigma^\text{P}_2[k*]$~\cite{de2017parameterized}.

\section{Conclusion}\label{sec:sep:conclusion}

In this work we studied the computational complexity of deciding whether a temporal graph admits separators of bounded size under the restless temporal path model studied by \citet{HMZ19}. We established that \probRSeparator is complete for $\Sigma^\text{P}_2$, a complexity class that is located in the second level of the polynomial time hierarchy. This implies, for example, that we presumably cannot use SAT-solvers or ILP-solvers to compute restless temporal separators. We also provide some preliminary results for the parameterized complexity of \probRSeparator parameterized by the separator size~$k$. We show that the parameterized problem is hard for para-coNP and hard for \wtwo. We conjecture that and leave as an open question for future reseach whether the parameterized problem is complete for~$\Sigma^\text{P}_2[k*]$~\cite{de2017parameterized}.

\bibliographystyle{plainnat}

\bibliography{bibliography}

\newcommand{\noopsort}[1]{}
\begin{thebibliography}{28}
\providecommand{\natexlab}[1]{#1}
\providecommand{\url}[1]{\texttt{#1}}
\expandafter\ifx\csname urlstyle\endcsname\relax
  \providecommand{\doi}[1]{doi: #1}\else
  \providecommand{\doi}{doi: \begingroup \urlstyle{rm}\Url}\fi

\bibitem[Ahuja et~al.(1993)Ahuja, Magnanti, and Orlin]{AMO93}
Ravindra~K. Ahuja, Thomas~L. Magnanti, and James~B. Orlin.
\newblock \emph{{Network Flows: Theory, Algorithms and Applications}}.
\newblock Prentice Hall, 1993.

\bibitem[Arora and Barak(2009)]{AB09}
Sanjeev Arora and Boaz Barak.
\newblock \emph{Computational Complexity: A Modern Approach}.
\newblock Cambridge University Press, 2009.

\bibitem[Barab{\'a}si(2016)]{Bar16}
Albert-L{\'a}szl{\'o} Barab{\'a}si.
\newblock \emph{Network Science}.
\newblock Cambridge University Press, 2016.

\bibitem[Bentert et~al.(2020)Bentert, Himmel, Nichterlein, and
  Niedermeier]{Him+20}
Matthias Bentert, Anne-Sophie Himmel, Andr{\'e} Nichterlein, and Rolf
  Niedermeier.
\newblock Efficient computation of optimal temporal walks under waiting-time
  constraints.
\newblock \emph{Applied Network Science}, 5\penalty0 (1):\penalty0 1--26, 2020.

\bibitem[Berman(1996)]{Ber96}
Kenneth~A Berman.
\newblock Vulnerability of scheduled networks and a generalization of
  {{M}}enger's theorem.
\newblock \emph{Networks: An International Journal}, 28\penalty0 (3):\penalty0
  125--134, 1996.

\bibitem[Casteigts et~al.(2012)Casteigts, Flocchini, Quattrociocchi, and
  Santoro]{Cas+12}
Arnaud Casteigts, Paola Flocchini, Walter Quattrociocchi, and Nicola Santoro.
\newblock Time-varying graphs and dynamic networks.
\newblock \emph{International Journal of Parallel, Emergent and Distributed
  Systems}, 27\penalty0 (5):\penalty0 387--408, 2012.

\bibitem[Casteigts et~al.(2020)Casteigts, Himmel, Molter, and Zschoche]{HMZ19}
Arnaud Casteigts, Anne-Sophie Himmel, Hendrik Molter, and Philipp Zschoche.
\newblock The computational complexity of finding temporal paths under waiting
  time constraints.
\newblock In \emph{Proceedings of the 31st International Symposium on
  Algorithms and Computation ({ISAAC} '20)}, volume 181 of \emph{LIPIcs}, pages
  30:1--30:18. Schloss Dagstuhl--Leibniz-Zentrum f{\"u}r Informatik, 2020.

\bibitem[Cygan et~al.(2015)Cygan, Fomin, Kowalik, Lokshtanov, Marx, Pilipczuk,
  Pilipczuk, and Saurabh]{Cyg+15}
Marek Cygan, Fedor~V. Fomin, {\L{}}ukasz Kowalik, Daniel Lokshtanov,
  D{\'{a}}niel Marx, Marcin Pilipczuk, Micha\l{} Pilipczuk, and Saket Saurabh.
\newblock \emph{Parameterized Algorithms}.
\newblock Springer, 2015.

\bibitem[Downey and Fellows(1999)]{DF99}
Rodney~G Downey and Michael~R Fellows.
\newblock \emph{Parameterized Complexity}.
\newblock Springer, 1999.

\bibitem[Downey and Fellows(2013)]{DF13}
Rodney~G Downey and Michael~R Fellows.
\newblock \emph{Fundamentals of Parameterized Complexity}.
\newblock Springer, 2013.

\bibitem[Flum and Grohe(2006)]{FG06}
J{\"o}rg Flum and Martin Grohe.
\newblock \emph{Parameterized Complexity Theory}, volume XIV of \emph{Texts in
  Theoretical Computer Science. An EATCS Series}.
\newblock Springer, 2006.

\bibitem[Fluschnik et~al.(2020)Fluschnik, Molter, Niedermeier, Renken, and
  Zschoche]{Flu+19a}
Till Fluschnik, Hendrik Molter, Rolf Niedermeier, Malte Renken, and Philipp
  Zschoche.
\newblock Temporal graph classes: A view through temporal separators.
\newblock \emph{Theoretical Computer Science}, 806:\penalty0 197--218, 2020.

\bibitem[Haan and Szeider(2017)]{de2017parameterized}
Ronald~{\noopsort{Haan}}de Haan and Stefan Szeider.
\newblock Parameterized complexity classes beyond para-{NP}.
\newblock \emph{Journal of Computer and System Sciences}, 87:\penalty0 16--57,
  2017.

\bibitem[Holme(2015)]{Hol15}
Petter Holme.
\newblock Modern temporal network theory: a colloquium.
\newblock \emph{The European Physical Journal B}, 88\penalty0 (9):\penalty0
  234:1--234:30, 2015.

\bibitem[Holme(2016)]{Hol16}
Petter Holme.
\newblock Temporal network structures controlling disease spreading.
\newblock \emph{Physical Review E}, 94.2:\penalty0 022305:1--022305:8, 2016.

\bibitem[Holme and Saram\"{a}ki(2019)]{HS19}
Petter Holme and Jari Saram\"{a}ki.
\newblock \emph{Temporal Network Theory}.
\newblock Springer, 2019.

\bibitem[Kempe et~al.(2002)Kempe, Kleinberg, and Kumar]{KKK02}
David Kempe, Jon Kleinberg, and Amit Kumar.
\newblock Connectivity and inference problems for temporal networks.
\newblock \emph{Journal of Computer and System Sciences}, 64\penalty0
  (4):\penalty0 820--842, 2002.

\bibitem[Kermack and McKendrick(1927)]{kermack1927contribution}
William~Ogilvy Kermack and Anderson~G McKendrick.
\newblock A contribution to the mathematical theory of epidemics.
\newblock \emph{Proceedings of the Royal Society of London, Series A.},
  115\penalty0 (772):\penalty0 700--721, 1927.

\bibitem[Latapy et~al.(2018)Latapy, Viard, and Magnien]{LVM18}
Matthieu Latapy, Tiphaine Viard, and Cl{\'e}mence Magnien.
\newblock Stream graphs and link streams for the modeling of interactions over
  time.
\newblock \emph{Social Network Analysis and Mining}, 8\penalty0 (1):\penalty0
  61:1--61:29, 2018.

\bibitem[Maack et~al.(2021)Maack, Molter, Niedermeier, and Renken]{maack21}
Nicolas Maack, Hendrik Molter, Rolf Niedermeier, and Malte Renken.
\newblock On finding separators in temporal split and permutation graphs.
\newblock \emph{arXiv preprint arXiv:2105.12003}, 2021.

\bibitem[Michail(2016)]{Mic16}
Othon Michail.
\newblock An introduction to temporal graphs: An algorithmic perspective.
\newblock \emph{Internet Mathematics}, 12\penalty0 (4):\penalty0 239--280,
  2016.

\bibitem[Molter(2020)]{Molter20}
Hendrik Molter.
\newblock \emph{Classic Graph Problems Made Temporal - A Parameterized
  Complexity Analysis}.
\newblock Phd thesis, Technische Universität Berlin, December 2020.
\newblock URL \url{http://dx.doi.org/10.14279/depositonce-10551}.

\bibitem[Newman(2018)]{New18}
Mark E~J Newman.
\newblock \emph{Networks}.
\newblock Oxford University Press, 2018.

\bibitem[Niedermeier(2006)]{Nie06}
Rolf Niedermeier.
\newblock \emph{Invitation to Fixed-Parameter Algorithms}.
\newblock Oxford University Press, 2006.

\bibitem[Paz and Moran(1981)]{PM81}
Azaria Paz and Shlomo Moran.
\newblock Non deterministic polynomial optimization problems and their
  approximations.
\newblock \emph{Theoretical Computer Science}, 15\penalty0 (3):\penalty0
  251--277, 1981.

\bibitem[Stockmeyer(1976)]{Sto76}
Larry~J Stockmeyer.
\newblock The polynomial-time hierarchy.
\newblock \emph{Theoretical Computer Science}, 3\penalty0 (1):\penalty0 1--22,
  1976.

\bibitem[Zschoche(2017)]{Zsc17}
Philipp Zschoche.
\newblock On finding separators in temporal graphs, August 2017.
\newblock Master's thesis.

\bibitem[Zschoche et~al.(2020)Zschoche, Fluschnik, Molter, and
  Niedermeier]{Zsc+19}
Philipp Zschoche, Till Fluschnik, Hendrik Molter, and Rolf Niedermeier.
\newblock The complexity of finding separators in temporal graphs.
\newblock \emph{Journal of Computer and System Sciences}, 107:\penalty0 72--92,
  2020.

\end{thebibliography}

\end{document}